\documentclass[conference, a4paper]{IEEEtran}

\usepackage{setspace, dsfont, enumerate}
\usepackage{amssymb, amsmath, mathrsfs, amsthm, stmaryrd, upgreek, mathtools}
\usepackage{pgfplots, caption, tikz, subfigure, graphicx}
\usepackage{url, hyperref, cite}
\usepackage{MyMnSymbol}

\AtBeginDocument{
  \addtolength\abovedisplayskip{-2.2pt}
  \addtolength\belowdisplayskip{-2.2pt}
}

\input{def.tex}

\begin{document}

\title{Density Criteria for the Identification of Linear Time-Varying Systems}

\author{\IEEEauthorblockN{C\'eline Aubel and Helmut B\"olcskei \medskip}
\IEEEauthorblockA{Dept.~IT~\&~EE, ETH Zurich, Switzerland\\
Email: \{aubelc, boelcskei\}@nari.ee.ethz.ch}
}

\maketitle

\newcommand{\doppler}{\nu}
\newcommand{\delay}{\tau}
\newcommand{\idx}{k}
\newcommand{\idxSet}{\Omega}
\newcommand{\meas}{\mu}
\newcommand{\norm}[1]{\left\|#1\right\|}
\newcommand{\bargmann}{\mathcal{B}}
\newcommand{\bargmannSpace}{\mathcal{F}(\C)}
\newcommand{\realFunc}{f}
\newcommand{\complexFuncOne}{\varphi}
\newcommand{\gaussian}{g}
\newcommand{\sig}{x}
\newcommand{\STFT}{\mathcal{V}_\gaussian}
\newcommand{\linOp}{\mathcal{H}}
\newcommand{\linOpBis}{\mathcal{K}}
\newcommand{\linOpSpace}[1]{\mathscr{H}(#1)}
\newcommand{\discreteSet}{\Xi}

\linespread{1}

\begin{abstract} \boldmath
	This paper addresses the problem of identifying a linear time-varying (LTV) system characterized by a (possibly infinite) discrete set of delays and Doppler shifts. We prove that stable identifiability is possible if the upper uniform Beurling density of the delay-Doppler support set is strictly smaller than $1/2$ and stable identifiability is impossible for densities strictly larger than $1/2$. The proof of this density theorem reveals an interesting relation between LTV system identification and interpolation in the Bargmann-Fock space. Finally, we introduce a subspace method for solving the system identification problem at hand.
\end{abstract}

\section{Introduction}

Identification of deterministic linear time-varying (LTV) systems has been a topic of long-standing interest, dating back to the seminal work by Kailath~\cite{Kailath1962} and Bello~\cite{Bello1969}, and has seen significant renewed interest during the past decade~\cite{Kozek2005,Heckel2013,Bajwa2011,Heckel2014}.
The problem arises in numerous application areas such as radar imaging and wireless communication. The formal problem statement is as follows. We want to identify the LTV system $\linOp$ from its response
\begin{equation}
	\forall t \in \R, \quad (\linOp\sig)(t) \triangleq \sum_{k \in \idxSet} a_k \sig(t - \delay_k)e^{-2\pi i\doppler_k t}
	\label{eq: radio communication channel}
\end{equation}
to the probing signal $\sig$, where $(\delay_k, \doppler_k)$ are delay-Doppler shift parameters, $a_k$ are the corresponding complex weights, and $\idxSet$ is a possibly infinite index set.

\textit{Contributions:} The purpose of this paper is to establish fundamental limits on the stable identifiability of $\linOp$ in \eqref{eq: radio communication channel} in terms of $\{a_k, \delay_k, \doppler_k\}_{k \in \idxSet}$. Our approach is based on the following insight. Defining the complex discrete measure $\mu_\linOp \triangleq \sum_{k \in \idxSet} a_k \delta_{\delay_k, \doppler_k}$ on $\R^2$, the input-output relation \eqref{eq: radio communication channel} can be rewritten as
\begin{equation*}
	\forall t \in \R, \quad (\linOp\sig)(t) = \int_{\R^2} \sig(t - \delay)e^{-2\pi i\doppler t}\dmeas{\meas_\linOp}{\delay, \doppler}.
\end{equation*}
Identifying the system $\linOp$ thus amounts to recovering the discrete measure $\mu_\linOp$ from $(\linOp\sig)(t)$, $t \in \R$. 
This formulation reveals an interesting connection to the super-resolution problem as studied by Donoho~\cite{Donoho1991}, where the goal is to recover a complex discrete measure on $\R$ (a weighted Dirac train) from lowpass measurements. The problem at hand, albeit formally similar, differs in two important aspects. First, we want to recover a measure $\mu_\linOp$ on $\R^2$, that is, a measure on a \emph{two-dimensional} set, from observations in \emph{one} parameter, namely $(\linOp\sig)(t)$, $t \in \R$. Second, the lowpass observations in~\cite{Donoho1991} are replaced by short-time Fourier transform-type observations, where the probing signal $\sig$ appears as the window function. These differences make for quite different technical challenges. Nevertheless, we can follow the spirit of Donoho's work~\cite{Donoho1991}, who established necessary and sufficient conditions for unique recovery in the super-resolution problem. These conditions are expressed in terms of the uniform Beurling densities of the measure's support set and are derived using density theorems for interpolation in the Bernstein and Paley-Wiener spaces~\cite{Beurling1989-2} and for balayage of Fourier-Stieltjes transforms~\cite{Beurling1989-1}.
Here, we will use a density theorem for interpolation in the Bargmann-Fock space~\cite{Seip1992, Seip1992-1, Seip1992-2, Seip1992-3}. Our main result says that stable identifiability is possible if the upper uniform Beurling density of the set $\{(\delay_k, \doppler_k)\}_{k \in \idxSet}$ is strictly smaller than $1/2$ and stable identifiability is impossible for densities strictly larger than $1/2$. Moreover, we present a subspace method for recovering the parameters $\{a_k, \delay_k, \doppler_k\}_{k \in \idxSet}$ from $\linOp\sig$ when $\sig$ is a Gaussian signal.

\textit{Relation to prior work:} Kozek and Pfander showed in~\cite{Kozek2005} that Gabor frame operators on rectangular lattices $a^{-1}\Z \times b^{-1}\Z$ are identifiable if and only if $ab \leq 1$. This problem is equivalent to the LTV system identification problem considered here for the pairs $(\delay_k, \doppler_k)$ lying on the lattice $a^{-1}\Z \times b^{-1}\Z$. The identifiability condition in \cite{Kozek2005} holds irrespectively of whether the set $\{(\delay_k, \doppler_k)\}_{k \in \idxSet}$ occupies the entire lattice or just parts of it, i.e., the result does not take into account the density of $\{(\delay_k, \doppler_k)\}_{k \in \idxSet}$ across $\R^2$. The results in \cite{Heckel2013} pertain to the identification of LTV systems with compactly supported spreading function and do not encompass operators defined by~\eqref{eq: radio communication channel}.
In \cite{Grip2013}, a necessary condition for identifiability of a set of Hilbert-Schmidt operators defined by atomic decompositions is given; the condition is expressed in terms of a ``2-dimensional'' Beurling density, but the operator class considered does not contain the operators characterized by \eqref{eq: radio communication channel}. In \cite{Bajwa2011}, it is shown that parametric underspread LTV systems, i.e., LTV systems with delay-Doppler spread product smaller than $1$, can be identified provided that the time-bandwidth product of the probing signal is large enough. In~\cite{Heckel2014}, a convex programming algorithm for stable recovery of the triplets $\{a_k, \delay_k, \doppler_k\}_{k \in \idxSet}$ from a finite number of noisy samples of $\linOp\sig$ is proposed; the algorithm assumes that the points $\{(\delay_k, \doppler_k)\}_{k \in \idxSet}$ obey a minimum separation condition.

\textit{Notation:} The complex conjugate of $z \in \C$ is denoted by~$\overline{z}$. For a Hilbert space $H$, we write $\innerProd{\cdot}{\cdot}_H$ and $\norm{\cdot}_H$ for the inner product and norm on $H$, respectively. 
Linear operators are denoted by uppercase calligraphic letters. For $\doppler \in \R$, we define the modulation operator $(\mathcal{M}_\doppler\sig)(t) \triangleq \sig(t)e^{-2\pi i\doppler t}$, and for $\delay \in \R$, the translation operator $(\mathcal{T}_\delay\sig)(t) \triangleq \sig(t - \delay)$. For a measure space $(X, \Sigma, \mu)$ and a measurable function $\varphi \colon X \rightarrow  \C$, we write $\int_X \varphi(x) \dmeas{\mu}{x}$ for the integral of $\varphi$ with respect to $\mu$, and we set $\mathrm{d}x \triangleq \dmeas{\lambda}{x}$ if $\lambda$ is the Lebesgue measure.  
If $X$ is a topological space, $\mathcal{B}(X)$ designates the Borel $\sigma$-algebra over $X$ and $\measSpace{X}$ is the space of all complex Radon measures on $(X, \mathcal{B}(X))$. For $x \in X$, $\delta_x \in \measSpace{X}$ denotes the Dirac measure at $x$, which for $B \in \mathcal{B}(X)$ is given by $\delta_x(B) = 1$, if $x \in B$, and $\delta_x(B) = 0$, else. The support $\supp(\mu)$ of a complex Radon measure $\mu \in \measSpace{X}$ is the largest closed set $C \subseteq X$ such that for every open set $B \in \mathcal{B}(X)$ satisfying $B \cap C \neq \emptyset$, it holds that $\mu(B \cap C) \neq 0$. We define the total variation (TV) norm $\normTV{\cdot}$ on $\measSpace{X}$ as $\normTV{\mu} \triangleq \abs{\mu}\!(X)$, where $\abs{\mu} \in \measSpace{X}$ is the total variation of $\mu$ given by
\begin{equation*}
	\forall B \in \mathcal{B}(X), \quad \abs{\mu}\!(B) \triangleq \sup_{\pi \in \Pi(B)} \sum_{A \in \pi} \abs{\mu(A)}
\end{equation*}
with $\Pi(B)$ denoting the set of all partitions of $B$. For a discrete measure $\mu = \sum_{k \in \idxSet}  \alpha_k\delta_{x_k} \in \measSpace{X}$, where $\{x_k\}_{k \in \idxSet}$ and $\{\alpha_k\}_{k \in \idxSet}$ are sequences in $X$ and $\C$, respectively, we have $\normTV{\mu} = \sum_{k \in \idxSet} \abs{\alpha_k}$, and we define the norm $\norm{\mu}_2 \triangleq \sqrt{\sum_{k \in \idxSet} \abs{\alpha_k}^2}$. $S_0(\R)$ stands for the Feichtinger algebra~\cite{Groechenig2000}. A set $\Lambda \subseteq \R^2$ is said to be discrete if for all $\lambda \in \Lambda$ one can find $\delta > 0$ such that $\norm{\lambda - \lambda'}_{\ell_2} > \delta$ for all $\lambda' \in \Lambda \!\setminus \{\lambda\}$. A set $\Lambda \subseteq \R^2$ is said to be uniformly discrete if $\inf\{\norm{\lambda - \lambda'}_{\ell^2} \colon \lambda, \lambda' \in \Lambda, \lambda \neq \lambda'\} > 0$.

\section{Problem formulation}

Throughout the paper, we let $\discreteSet \subseteq A\Z^2$ be a lattice, where $A \in \R^{2\times 2}$ is an invertible matrix.
We consider LTV systems characterized by the bounded linear operator $\linOp \colon S_0(\R) \rightarrow S_0(\R)$ defined as 
\begin{equation*}
	\forall t \in \R, \quad (\linOp\sig)(t) \triangleq \int_{\R^2} \sig(t - \delay)e^{-2\pi i\doppler t}\dmeas{\meas_\linOp}{\delay, \doppler},
\end{equation*}
where $\meas_\linOp \in \measSpace{\discreteSet}$. The vector space of all such operators is denoted by $\linOpSpace{\discreteSet}$ and is equipped with the norm $\norm{\linOp}_{\mathscr{H}} \triangleq \norm{\meas_\linOp}_2$. 
We pursue two principal goals. First, we want to establish conditions under which $\linOp \in \linOpSpace{\discreteSet}$ is identifiable, that is, one can find a signal $\sig \in S_0(\R)$, henceforth called probing signal, such that $\linOp\sig$ uniquely determines $\linOp$. The second goal is to find a method for recovering the triplets $\{a_k, \delay_k$, $\doppler_k\}_{k \in \idxSet}$ associated with $\linOp$ from the reponse of $\linOp$ to a probing signal. Throughout the paper, we consider stable identifiability (hereafter simply referred to as identifiability) which guarantees that $\linOp$ can not only be recovered from $\linOp\sig$, but small errors in $\linOp\sig$ also result in small errors in the identified operator. 
\begin{defn}[Stable identifiability]
	Let $\mathscr{I} \subseteq \linOpSpace{\discreteSet}$ be a set of operators. We say that $\mathscr{I}$ is identifiable if there exist $\sig \in S_0(\R)$ and constants $C_1, C_2$, $0 < C_1 \leq C_2 < \infty$, such that
	\begin{equation}
		C_1\norm{\linOp - \linOpBis}_\mathscr{H} \leq \norm{\linOp\sig - \linOpBis\sig}_{L^2} \leq C_2\norm{\linOp - \linOpBis}_\mathscr{H}
		\label{eq: identifiability set of operators}
	\end{equation}
	 for all $\linOp, \linOpBis \in \mathscr{I}$.
\end{defn}

The upper bound in~\eqref{eq: identifiability set of operators} is met trivially thanks to the following result.
\begin{prop}
	\label{prop: upper bound definition identifiability set of operator}
	Let $\sig \in S_0(\R)$.
	For all $\linOp, \linOpBis \in \linOpSpace{\discreteSet}$, it holds that \vspace{-0.1cm}
	\begin{equation*}
		\norm{\linOp\sig - \linOpBis\sig}_{L^2} \leq \norm{\linOp - \linOpBis}_\mathscr{H}\norm{\sig}_{L^2} .
	\end{equation*}
\end{prop}

\begin{proof}
	Follows directly by application of \cite[Thm.~12]{Heil2007} as $x \in S_0(\R)$ and $\supp(\mu_\linOp - \mu_\linOpBis) \subseteq \discreteSet$ is uniformly discrete.
\end{proof}

Proving that a set of operators $\mathscr{I} \subseteq \linOpSpace{\discreteSet}$ is identifiable therefore amounts to finding a probing signal $x \in S_0(\R)$ such that the lower bound in~\eqref{eq: identifiability set of operators} holds for all $\linOp, \linOpBis \in \mathscr{I}$.

\section{Main result}

We derive identifiability conditions for operators $\linOp \in \mathscr{H}(\discreteSet)$ in terms of the uniform Beurling density of $\supp(\mu_\linOp) = \{(\delay_\idx, \doppler_\idx)\}_{\idx \in \idxSet}$. The uniform Beurling density measures the average number of points in $\supp(\mu_\linOp)$ per unit cell of $\R^2$. We will actually have to work with lower and upper uniform Beurling densities.

\begin{defn}[Lower and upper uniform Beurling densities, {\cite[p.~346]{Beurling1989-1}\cite[p.~47]{Landau1966}}]
	Let $\Lambda$ be a uniformly discrete set in $\R^2$. Let $Q \subseteq \R^2$ be a compact set of measure $\lambda(Q) =1$ whose boundary has measure $0$. For $r >0$, let $n^-(\Lambda, rQ)$ and $n^+(\Lambda, rQ)$ be, respectively, the smallest and largest number of points of $\Lambda$ contained in any translate of~$rQ$. The quantities
	\begin{equation*}
		\liminf_{r \rightarrow \infty} \frac{n^-(\Lambda, rQ)}{r^2} \qquad \text{ and } \qquad \limsup_{r \rightarrow \infty} \frac{n^+(\Lambda, rQ)}{r^2}
	\end{equation*}
	do not depend on $Q$ and define the lower and upper uniform Beurling densities of $\Lambda$, respectively. These quantities are denoted by $D^-(\Lambda)$ and $D^+(\Lambda)$, respectively. If $D^-(\Lambda) = D^+(\Lambda)$, then $\Lambda$ is said to have uniform Beurling density $D(\Lambda) \triangleq D^-(\Lambda) = D^+(\Lambda)$.
\end{defn}

In the remainder of the paper, we will often deal with sets of complex numbers $\{z_k\}_{k \in \idxSet}$, whose lower and upper uniform Beurling densities we define to be the lower and upper uniform Beurling densities of the sets $\left\{(\Re{z_\idx}\!, \Im{z_\idx})\right\}_{\idx \in \idxSet}$ in $\R^2$.

We are now ready to state our main result.

\begin{thm}[Density criteria for identifiability]
	Let $\alpha$ be a positive number and assume that $\discreteSet$ has uniform Beurling density $D(\discreteSet) = 1/\det(A) \geq 2\alpha$.
	Define the set of operators $\mathscr{H}_\alpha(\discreteSet) \triangleq \left\{\linOp \in \linOpSpace{\discreteSet} \colon D^+(\supp(\meas_\linOp)) \leq \alpha\right\}$ and let $\gaussian(t) \triangleq \sqrt{B}e^{-\pi B^2t^2/2}$, $t \in \R$, where $B > 0$. 
	\begin{enumerate}[\hspace{1cm}a)]
		\item If $\alpha < 1/2$, then $\mathscr{H}_\alpha(\discreteSet)$ is identifiable by $\gaussian$. \label{enum: identifiability}
		\item If $\alpha > 1/2$, then $\mathscr{H}_\alpha(\discreteSet)$ is not identifiable. \label{enum: non-identifiability}
	\end{enumerate}
	\label{thm: identifiability condition}
\end{thm}

Theorem~\ref{thm: identifiability condition} says that operators $\linOp \in \mathscr{H}(\discreteSet)$ are identifiable if they are ``sparse enough'' in the sense of the upper uniform Beurling density of $\supp(\mu_\linOp)$ satisfying $D^+(\supp(\mu_\linOp)) < 1/2$.
The proof of Theorem \ref{thm: identifiability condition} is given in Section~\ref{sec: proofs} and relies on the theory of interpolation in the Bargmann-Fock space $\bargmannSpace$ of entire functions $\complexFuncOne$ for which
\begin{equation*}
	\norm{\complexFuncOne}_\mathcal{F} \triangleq \left(\int_\C \abs{\complexFuncOne(z)}^2 e^{-\pi \abs{z}^2} \mathrm{d}z\right)^{1/2} < \infty.
\end{equation*}
This powerful theory was developed by Seip and Brekke in~\cite{Seip1992, Seip1992-1, Seip1992-2, Seip1992-3} and led to strong results in Gabor theory establishing density criteria on sets $\Lambda \triangleq \{(\delay_k, \doppler_k)\}_{k \in \idxSet}$  for $\{\mathcal{M}_{\doppler_k}\mathcal{T}_{\delay_k}g\}_{k \in \idxSet}$ to form a Riesz sequence in $L^2(\R)$. 

The result that comes closest to our Theorem~\ref{thm: identifiability condition} is \cite[Thm.~4.1]{Kozek2005} summarized next. Consider the set of Gabor frame operators $\mathscr{S}_{a, b} \triangleq \left\{\mathcal{S}^{a, b}_{g, h} \colon g \in L^2(\R), h \in W(\R)\right\}$, where $a, b > 0$, $W(\R)$ is the Wiener space
\begin{equation*}
	W(\R) \triangleq \left\{f \in L^2(\R) \colon \sum_{k \in \Z} \norm{f \cdot 1_{[k, k+1)}}_{L^\infty} < \infty \right\},
\end{equation*}
and $\mathcal{S}^{a, b}_{g, h}$ is the linear operator mapping $f \in L^2(\R)$ to
\begin{equation*}
	\mathcal{S}^{a, b}_{g, h}f \triangleq (ab)^{-1}\sum_{m, n \in \Z} \innerProd{h}{\mathcal{M}_{m/a}\mathcal{T}_{n/b}g}\mathcal{M}_{m/a}\mathcal{T}_{n/b}f
\end{equation*}
in $L^2(\R)$. $\mathscr{S}_{a, b}$ is shown in \cite{Kozek2005} to be identifiable if and only if $ab \leq 1$. This statement is a universal result in the sense of holding irrespectively of whether the set of delay-Doppler shift pairs $\{(\delay_k, \doppler_k)\}_{\idx \in \idxSet}$ occupies the entire lattice $a^{-1}\Z \times b^{-1}\Z$ or just parts of it. In contrast, our result takes into account the lattice occupation density. The factor of two difference between \cite[Thm.~4.1]{Kozek2005} and Theorem~\ref{thm: identifiability condition} in the critical density stems from the fact that for $\linOp, \linOpBis \in \mathscr{H}_\alpha(\discreteSet)$, we have $\linOp - \linOpBis \in \mathscr{H}_{2\alpha}(\discreteSet)$, whereas the set of operators $\mathscr{S}_{a, b}$ is linear, implying that for $\linOp, \linOpBis \in \mathscr{S}_{a, b}$, $\linOp - \linOpBis \in \mathscr{S}_{a,b}$ as well.

\section{Proof of Theorem~\ref{thm: identifiability condition}}
\label{sec: proofs}

We first collect some basic facts about interpolation in the Bargmann-Fock space~\cite{Seip1992, Seip1992-1, Seip1992-2, Seip1992-3}.

\subsection{Preparatory material}

\begin{defn}[Set of interpolation,~\cite{Seip1992}]
	Let $\Gamma \triangleq \{z_\idx\}_{\idx \in \idxSet} \subseteq \C$ be a discrete set. If for every sequence $\{w_\idx\}_{\idx \in \idxSet} \in \ell^2(\idxSet)$ there exists a function $\complexFuncOne \in \bargmannSpace$ such that $e^{-\pi \abs{z_\idx}^2/2}\complexFuncOne(z_\idx) = w_\idx$, for all $\idx \in \idxSet$, then $\Gamma$ is said to be a set of interpolation for $\bargmannSpace$.
\end{defn}

\begin{thm}[Density theorem for interpolation in the \mbox{Bargmann}-Fock space, {\cite[Thm.~1.2]{Seip1992}, \cite[Thm.~1.2]{Seip1992-2}}]
	A discrete set $\Gamma \subseteq \C$ is a set of interpolation for $\bargmannSpace$ if and only if it is uniformly discrete and $D^+(\Gamma) < 1$.
	\label{thm: density theorem for interpolation}
\end{thm}

The Bargmann-Fock space is isomorphic to $L^2(\R)$, since the Bargmann transform defined for $\realFunc \in L^2(\R)$ as
\begin{equation*}
	\forall z \in \C, \quad (\bargmann\realFunc)(z) \triangleq 2^{1/4}e^{-\pi z^2/2}\int_{-\infty}^\infty \realFunc(u) e^{2\pi uz - \pi u^2}\mathrm{d}u
\end{equation*}
is an isometric isomorphism from $L^2(\R)$ to $\bargmannSpace$. The Bargmann transform is closely related to the short-time Fourier transform (STFT) with respect to a Gaussian window~\cite[Def.~3.4.1]{Groechenig2000}. Specifically, the STFT of $\sig \in L^2(\R)$ with Gaussian window function $g(t) \triangleq  \sqrt{B}e^{-\pi B^2t^2/2}$, $t \in \R$, defined as
\begin{equation*}
	\forall (\delay, \doppler) \in \R^2, \quad (\mathcal{V}_g\sig)(\delay, \doppler) \triangleq \int_{-\infty}^\infty \sig(t)g(t - \delay)e^{-2\pi i\doppler t}\mathrm{d}t
\end{equation*}
can be expressed in terms of the Bargmann transform as
\begin{equation}
	\forall (\delay, \doppler) \in \R^2, \quad (\STFT\sig)(\delay, \doppler) = (\bargmann f)(z) e^{-\pi \abs{z}^2/2}e^{-\pi i\delay\doppler},
	\label{eq: relation STFT Bargmann transform}
\end{equation}
where we set $f(u) \triangleq (\sqrt{2}/B)^{1/2} \sig(u\sqrt{2}/B)$ and $z \triangleq \delay B/\sqrt{2} - i\doppler \sqrt{2}/B$.

\begin{defn}[Riesz sequence, {\cite[Chap.~3]{Christensen2008}}]
	Let $\{f_n\}_{n \in \Z}$ be a sequence of functions in $L^2(\R)$. 
	The sequence $\{f_n\}_{n \in \Z}$ is a Riesz sequence in $L^2(\R)$ if and only if there exist constants $C_1, C_2$, $0 < C_1 \leq C_2 < \infty$, such that
	\begin{equation*}
		C_1\sqrt{\sum_{n \in \Z} \abs{c_n}^2} \leq \norm{\sum_{n \in \Z} c_n f_n}_{L^2}  \leq C_2\sqrt{\sum_{n \in \Z} \abs{c_n}^2}
	\end{equation*}
	for all $\{c_n\}_{n \in \Z} \in \ell^2(\Z)$. 
\end{defn}

An important consequence of the relation between the STFT and the Bargmann transform is the following: the set $\{\mathcal{M}_{\doppler_\idx}\mathcal{T}_{\delay_\idx}\gaussian\}_{\idx \in \idxSet}$ forms a Riesz sequence in $L^2(\R)$ if and only if the set $\Gamma \triangleq \{z_\idx\}_{\idx \in \idxSet}$, where $z_\idx \triangleq \delay_\idx B/\sqrt{2} - i\doppler_\idx \sqrt{2}/B$, for all $\idx \in \idxSet$, is a set of interpolation for $\bargmannSpace$.

\subsection{Proof of Theorem~\ref{thm: identifiability condition}, Statement \ref{enum: identifiability})}
\label{sec: proof theorem identifiability condition}

Our proof is inspired by the technique used in~\cite[Thm.~1.1]{Donoho1991} to prove that discrete complex measures of the form $\mu = \sum_{k \in \Z} a_k\delta_{kT}$ are uniquely characterized by their Fourier transform $\hat{\mu}(f) = \sum_{k \in \Z} a_k e^{-2\pi ikTf}$ for $\abs{f} \leq f_c$ if the support of $\mu$ obeys $D^+(\supp(\mu)) < f_c$.

Let $\alpha < 1/2$ and $\linOp, \linOpBis \in \mathscr{H}_\alpha(\discreteSet)$. If $\linOp = \linOpBis$, then \eqref{eq: identifiability set of operators} holds trivially. We consider the case $\linOp \neq \linOpBis$ in the following. The support of the measure $\eta \triangleq \meas_\linOp - \meas_\linOpBis$ is contained in $\Lambda \triangleq \supp(\meas_\linOp) \hspace{0.5pt}\ensuremath \raisebox{-0.5pt}{$\cup$}\supp(\meas_\linOpBis)$. Since $\Lambda \subseteq \discreteSet$ is uniformly discrete, we have
\begin{equation}
	D^+(\Lambda) \leq D^+(\supp(\meas_\linOp)) + D^+(\supp(\meas_\linOpBis)) \leq 2\alpha < 1.
	\label{eq: upper bound density union}
\end{equation}
We write $\Lambda = \{(\delay_\idx, \doppler_\idx)\}_{\idx \in \idxSet}$ and define the corresponding set $\Gamma \triangleq \{z_\idx\}_{\idx \in \idxSet}$, where $z_\idx \triangleq x_\idx + iy_\idx$, $x_\idx \triangleq \delay_\idx B/\sqrt{2}$, and $y_\idx \triangleq -\doppler_\idx \sqrt{2}/B$, for all $\idx \in \idxSet$. 
Let $Q \triangleq [0,1]^2$ and $Q_B \triangleq [0, \sqrt{2}/B]\times[0, B/\sqrt{2}]$ and note that $\lambda(Q) = \lambda(Q_B) = 1$ and the boundaries of $Q$ and $Q_B$ both have measure $0$. Moreover, for $r > 0$, $(x_\idx, y_\idx) \in rQ$ if and only if $(\delay_\idx, \doppler_\idx) \in rQ_B$. The largest number of points of $\Gamma$ found in any translate of $rQ$ then equals the largest number of points of $\Lambda$ found in any translate of $rQ_B$ for $r  > 0$.
Therefore, we have $D^+(\Gamma) = D^+(\Lambda)$. From \eqref{eq: upper bound density union}, it then follows by application of Theorem~\ref{thm: density theorem for interpolation} that there exists a function $\complexFuncOne \in \bargmannSpace$ solving the interpolation problem $e^{-\pi\abs{z_\idx}^2/2}\complexFuncOne(z_\idx) = w_\idx$, for all $\idx \in \idxSet$, where the sequence $\{w_\idx\}_{\idx \in \idxSet} \in \ell^2(\idxSet)$ is defined as
\begin{equation}
	\forall \idx \in \idxSet, \quad w_\idx = \frac{\overline{\alpha_\idx}}{\sqrt{\sum_{k \in \idxSet} \abs{\alpha_\idx}^2}} e^{\pi i\delay_{\idx_0}\doppler_{\idx_0}}, \vspace{-0.1cm}
	\label{eq: sequence wk}
\end{equation}
with $\alpha_\idx \triangleq \eta(\{(\delay_\idx, \doppler_\idx)\})$, for all $\idx \in \idxSet$. Note that $\alpha_k \neq 0$ for at least one $\idx \in \idxSet$, say $\idx_0 \in \idxSet$, since $\linOp \neq \linOpBis$. It therefore follows that $\sum_{\idx \in \idxSet} \abs{\alpha_\idx}^2 > 0$, which ensures that the $w_k$ in~\eqref{eq: sequence wk} are well-defined. 
By $\sum_{k \in \idxSet} \abs{w_k}^2 = 1 < \infty$, the sequence $\{w_\idx\}_{\idx \in \idxSet}$ is in $\ell^2(\idxSet)$.
Since the Bargmann transform $\bargmann \colon L^2(\R) \rightarrow \bargmannSpace$ is an isomorphism and $\varphi \in \bargmannSpace$, there exists an $\realFunc \in L^2(\R)$ such that $\complexFuncOne = \bargmann\realFunc$. Now, defining $\sig(t) \triangleq (B/\sqrt{2})^{1/2}\realFunc(tB/\sqrt{2})$, $t \in \R$, we make use of the relation \eqref{eq: relation STFT Bargmann transform} between the STFT and the Bargmann transform to write
\begin{equation*}
	\forall \idx \in \idxSet, \quad \complexFuncOne(z_\idx) = (\STFT\sig)(\delay_\idx, \doppler_\idx) e^{\pi \abs{z_\idx}^2/2} e^{\pi i\delay_\idx\doppler_\idx}.
\end{equation*}
It therefore follows that
\begin{equation}
	\forall \idx \in \idxSet, \quad (\STFT\sig)(\delay_\idx, \doppler_\idx) = \frac{\overline{\alpha_\idx}}{\sqrt{\sum_{k \in \idxSet} \abs{\alpha_k}^2}}.
	\label{eq: STFT at the points}
\end{equation}
By construction of $\sig$, it holds that
\begin{align*}
	\int_{\R^2} (\STFT\sig)(\delay, \doppler) \dmeas{\eta}{\delay, \doppler} &= \sum_{\idx \in \idxSet} \alpha_\idx(\STFT\sig)(\delay_\idx, \doppler_\idx) = \!\sqrt{\sum_{\idx \in \idxSet} \abs{\alpha_\idx}^2} \\
		&\hspace{-1cm}= \norm{\eta}_2 = \norm{\meas_\linOp - \meas_\linOpBis}_2 = \norm{\linOp - \linOpBis}_\mathscr{H},
\end{align*}
where the last equality is by the definition of $\norm{\cdot}_\mathscr{H}$.
On the other hand, we have
\begin{align}
	\int_{\R^2} &(\STFT\sig)(\delay, \doppler) \dmeas{\eta}{\delay, \doppler} \notag\\
		&= \int_{\R^2} \left(\int_{-\infty}^\infty \sig(t)\gaussian(t - \delay)e^{-2\pi i\doppler t}\mathrm{d}t\right) \dmeas{\eta}{\delay, \doppler} \notag \\
		&= \int_{-\infty}^\infty \sig(t)\left(\int_{\R^2} \gaussian(t - \delay)e^{-2\pi i \doppler t}\dmeas{\eta}{\delay, \doppler}\right) \mathrm{d}t \label{eq: proof application of fubini} \\
		&= \int_{-\infty}^\infty \sig(t)\Big((\linOp\gaussian)(t) - (\linOpBis\gaussian)(t)\Big) \mathrm{d}t \notag \\  
		&\leq \norm{\sig}_{L^2} \norm{\linOp\gaussian - \linOpBis\gaussian}_{L^2}, \label{eq: proof identifiability with gaussian Cauchy Schwarz}
\end{align}
where we used Fubini's theorem to get \eqref{eq: proof application of fubini}. The conditions for Fubini's theorem are satisfied as
\begin{align*}
	 \int_{\R^2}\left(\int_{-\infty}^\infty \abs{\sig(t)\gaussian(t - \delay)e^{-2\pi i\doppler t}}\mathrm{d}t\right) \dmeas{\abs{\eta}\!}{\delay, \doppler} \\
	 	\leq \norm{\sig}_{L^2}\norm{\gaussian}_{L^2}\normTV{\eta} < \infty.
\end{align*}
Finally, \eqref{eq: proof identifiability with gaussian Cauchy Schwarz} follows from the Cauchy-Schwarz inequality. Since $\sqrt{\sum_{k \in \idxSet} \abs{\alpha_k}^2} \leq \sum_{k \in \idxSet} \abs{\alpha_k} = \normTV{\eta}$ and $\normTV{\eta} < \infty$ by $\eta \in \measSpace{\discreteSet}$, we have $\sqrt{\sum_{k \in \idxSet} \abs{\alpha_k}^2} < \infty$. Combining this with $\alpha_{\idx_0} \neq 0$, it follows from \eqref{eq: STFT at the points} that $(\STFT\sig)(\delay_{\idx_0}, \doppler_{\idx_0}) \neq 0$. Since 
\begin{equation*}
	\forall (\delay, \doppler) \in \R^2, \quad (\STFT\sig)(\delay, \doppler) = \varphi(z)e^{-\pi\abs{z}^2/2}e^{-\pi i\delay\doppler},
\end{equation*}
where $z = \delay B/\sqrt{2} -i\doppler \sqrt{2}/B$, and since $\varphi$ is an entire function, $\STFT\sig$ is continuous. This implies that one can find a neighborhood $V$ of $(\delay_{k_0}, \doppler_{k_0})$ such that $(\STFT\sig)(\delay, \doppler) \neq 0$ for all $(\delay, \doppler) \in V$. Therefore, $\norm{\STFT\sig}_{L^2}^2 \geq \int_V \abs{(\STFT\sig)(\delay, \doppler)}^2 \mathrm{d}\delay\mathrm{d}\doppler > 0$. Since $\norm{\gaussian}_{L^2} = 1$, we have $\norm{\sig}_{L^2} = \norm{\STFT\sig}_{L^2} > 0$. This allows us to write \eqref{eq: proof identifiability with gaussian Cauchy Schwarz} in the form
\begin{equation*}
	C_1 \norm{\linOp - \linOpBis}_\mathscr{H} \leq \norm{\linOp\gaussian - \linOpBis\gaussian}_{L^2},
\end{equation*}
where $C_1 \triangleq 1/\norm{\sig}_{L^2} > 0$, thereby completing the proof.

\subsection{Proof of Theorem~\ref{thm: identifiability condition}, Statement~\ref{enum: non-identifiability})}
\label{sec: proof theorem non-identifiability condition with a gaussian}

By contraposition, we show that $\alpha \leq 1/2$ if $\mathscr{H}_\alpha(\discreteSet)$ is identifiable. To this end, let us assume that one can find a probing signal $x \in S_0(\R)$ and constants $C_1, C_2$, $0 < C_1 \leq C_2 < \infty$, such that \eqref{eq: identifiability set of operators} holds for all $\linOp, \linOpBis \in \mathscr{H}_\alpha(\discreteSet)$. We then construct sequences $\{\linOp_n\}_{n \in \N}$ and $\{\linOpBis_n\}_{n \in \N}$ of operators in $\mathscr{H}_\alpha(\discreteSet)$ that have $\Lambda_{\linOp_n} \triangleq \supp(\mu_{\linOp_n})$ and $\Lambda_{\linOpBis_n} \triangleq \supp(\mu_{\linOpBis_n})$ disjoint for all $n \in \N$, and show that if \eqref{eq: identifiability set of operators} is to hold for $\linOp = \linOp_n$ and $\linOpBis = \linOpBis_n$, for $n \rightarrow \infty$, we necessarily have $\alpha \leq 1/2$. Specifically, with $d \triangleq D(\discreteSet) = 1/\det(A)$, we construct $\{\Lambda_{\linOp_n}\}_{n \in \N}$ as follows. By density of $\mathbb{Q}$ in $\R$, we can find, for every $n \in \N$, a $q_n \in \mathbb{Q}$ such that
\begin{equation}
	\frac{\alpha}{d} - \frac{1}{n} \leq q_n \leq \frac{\alpha}{d}.
	\label{eq: sequence in Q}
\end{equation}
By assumption, $d \geq 2\alpha$, and hence $0 \leq q_n \leq 1/2$ for all $n \geq \lfloor \alpha/d\rfloor + 1 \reversetriangleq n_0$. We can therefore write $q_n = \ell_n/m_n$, where $m_n, \ell_n \in \N$ and $0 < \ell_n \leq m_n/2$ for all $n \geq n_0$. Next, define the set $\Lambda_{\linOp_n} \triangleq A\Delta_{\linOp_n}$, where 
\begin{equation*}
	\Delta_{\linOp_n} \triangleq \bigcup_{r = 0}^{\ell_n-1} (m_n\Z + r)^2
\end{equation*}
is a subset of $\Z^2$ with uniform Beurling density $D(\Delta_{\linOp_n}) = q_n$. The set $\Lambda_{\linOp_n}$ is then a subset of $\discreteSet$ with uniform Beurling density $D(\Lambda_{\linOp_n}) = q_n/\det(A) = q_nd \leq \alpha$ for all $n \geq n_0$. Moreover, it follows from \eqref{eq: sequence in Q} that $D(\Lambda_{\linOp_n})$ converges to $\alpha$ as $n \rightarrow \infty$. We next construct the sequence $\{\Lambda_{\linOpBis_n}\}_{n \in \N}$. Again invoking the density of $\mathbb{Q}$ in $\R$, we can find $q_n' \in \mathbb{Q}$ such that
\begin{equation}
	\frac{\alpha-1/n}{(1-q_n)d} \leq q_n' \leq \frac{\alpha}{(1-q_n)d},
	\label{eq: sequence in Q 2}
\end{equation}
for all $n \in \N$. As $d \geq 2\alpha$ and $q_n \leq 1/2$, we have $0 \leq q_n' \leq 1$ for all $n \geq \lfloor 1/\alpha\rfloor + 1 \reversetriangleq n_1$.
We define $\Lambda_{\linOpBis_n} \triangleq A\Delta_{\linOpBis_n}$, where  
\begin{equation*}
	\Delta_{\linOpBis_n} \triangleq \bigcup_{r = \ell_n}^{m_n-1}\bigcup_{s = 0}^{\ell_n'-1} (m_nm_n'\Z + m_ns + r)^2
\end{equation*}
is a subset of $\Z^2 \setminus \Delta_{\linOp_n}$ with uniform Beurling density $D(\Delta_{\linOpBis_n}) = q_n'(1-q_n)d \leq \alpha$. The set $\Lambda_{\linOpBis_n}$ is then a subset of $\discreteSet \setminus \Lambda_{\linOpBis_n}$ with uniform Beurling density $q_n'(1-q_n)/\det(A) = q_n'(1- q_n)d \leq \alpha$ for all $n \geq n_1$.
It follows from \eqref{eq: sequence in Q 2} that the sequence $D(\Lambda_{\linOpBis_n})$ converges to $\alpha$. Now let $n \geq \max\{n_0, n_1\}$. As $\Lambda_{\linOp_n}$ and $\Lambda_{\linOpBis_n}$ are disjoint and both have uniform Beurling density, the set $\Lambda_n \triangleq \Lambda_{\linOp_n} \cup \Lambda_{\linOpBis_n}$ has uniform Beurling density $D(\Lambda_n) = q_nd + q_n'(1 - q_n)d$. We write $\Lambda_n = \{(\delay_{k, n}, \doppler_{k, n})\}_{k \in \idxSet_n}$, $\idxSet_{\linOp_n} \triangleq \{k \in \idxSet_n \colon (\delay_{k, n}, \doppler_{k, n}) \in \Lambda_{\linOp_n}\}$, and $\idxSet_{\linOpBis_n} \triangleq \{k \in \idxSet_n \colon (\delay_{k, n}, \doppler_{k, n}) \in \Lambda_{\linOpBis_n}\}$.
Next, let $\{a_k\}_{k \in \idxSet_n}$ be a sequence in $\ell^1(\idxSet_n)$ and define the measures
\begin{equation*}
	\mu_{\linOp_n} \triangleq \sum_{k \in \idxSet_{\linOp_n}} \!\!a_k\delta_{\delay_{k, n}, \doppler_{k, n}}, \qquad  \mu_{\linOpBis_n} \triangleq -\sum_{k \in \idxSet_{\linOpBis_n}} \!\!a_k\delta_{\delay_{k, n}, \doppler_{k, n}}. 
\end{equation*}
The corresponding operators $\linOp_n$ and $\linOpBis_n$ are in $\mathscr{H}_\alpha(\discreteSet)$. Evaluating \eqref{eq: identifiability set of operators} with $\linOp = \linOp_n$ and $\linOpBis = \linOpBis_n$ then yields
\begin{equation*}
	C_1\sqrt{\sum_{k \in \idxSet_n} \!\!\abs{a_k}^2} \leq \norm{\sum_{k \in \idxSet_n} a_k \mathcal{M}_{\doppler_{k, n}}\mathcal{T}_{\delay_{k, n}}x}_{L^2} \!\!\!\!\leq C_2 \sqrt{\sum_{k \in \idxSet_n} \!\!\abs{a_k}^2}.
\end{equation*}
We can therefore conclude that $\{\mathcal{M}_{\doppler_{k, n}}\mathcal{T}_{\delay_{k, n}}x\}_{k \in \idxSet_n}$ is a Riesz sequence in $L^2(\R)$. By application of Theorem~\cite[Thm.~13-d)]{Heil2007}, it is then necessary that $D^+(\Lambda_n) \leq 1$, thus implying $q_nd + q_n'(1 - q_n)d \leq 1$.
Taking the limit $n \rightarrow \infty$ yields $2\alpha \leq 1$ and completes the proof.

\section{Identification algorithms}

In wireless or radar applications contributions to the received signal corresponding to distant reflectors and/or scatterers can often be neglected thanks to path loss.
It is therefore sensible to take the delays and Doppler shifts to lie on a compact set $[\delay_\mathrm{min}, \delay_\mathrm{max}] \times [\doppler_\mathrm{min}, \doppler_\mathrm{max}]$. Combined with the standing assumption of the $(\delay_k, \doppler_k)$ being supported on a lattice, the input-output relation~\eqref{eq: radio communication channel} reduces to
\begin{equation*}
	\forall t \in \R, \quad (\linOp\sig)(t) = \sum_{k = 1}^K a_k\sig(t - \delay_k)e^{-2\pi i \doppler_k t}
\end{equation*}
with $K$ finite.
This implies that the set $\supp(\mu_\linOp)$ is finite and hence $D^+(\supp(\mu_\linOp)) = 0$. Theorem~\ref{thm: identifiability condition} then tells us that $\linOp$ can be identified by the probing signal 
\begin{equation*}
	\forall t \in \R, \quad x(t) \triangleq \sqrt{B}\exp\!\left(-\frac{\pi B^2 (t - T)^2}{2}\right).
\end{equation*}
Furthermore, in practice, we only have access to samples of $\linOp\sig$ restricted to a finite time interval, say $[\delay_\mathrm{min}, \delay_\mathrm{min}+T]$. The identification problem therefore reduces to the identification of $\linOp$ from the samples $r_m \triangleq (\linOp\sig)(t_m)$ with $t_m \triangleq  \delay_\mathrm{min}+ mT/M$ for  $m \in \{0, 1, \ldots, M-1\}$. Straightforward manipulations reveal that
\begin{equation*}
	r_m = \lambda_{M-m}\sum_{k = 1}^K \alpha_kz_k^{M-m}, 
\end{equation*}
where we set
\begin{align*}
	&\lambda_{M-m} \triangleq \sqrt{B}\exp\!\left(-\frac{\pi B^2T^2(M-m)^2}{2M^2}\right) \\
	&\alpha_k \triangleq a_k \exp\!\left(-\frac{\pi B^2(\delay_k - \delay_\mathrm{min})^2}{2}\right)e^{-2\pi i\doppler_k T}e^{-2\pi i\doppler_k \delay_\mathrm{min}} \\
	&z_k \triangleq \exp\!\left(-\frac{\pi B^2T(\delay_k - \delay_\mathrm{min})}{M}\right) e^{2\pi i \doppler_k T/M}
\end{align*}
for $m \in \{0, 1, \ldots, M-1\}$. Therefore, letting $y_m \triangleq r_{M-m}/\lambda_{m}$, we obtain
\begin{equation*}
 	y_m = \sum_{k = 1}^K \alpha_k z_k^m
\end{equation*}
for all $m \in \{1, 2, \ldots, M\}$. We have therefore reduced the LTV system identification problem to the problem of identifying cisoids. A variety of subspace algorithms like ESPRIT \cite{Roy1986}, MUSIC~\cite{Schmidt1986}, or the matrix pencil method~\cite{Hua1990} can be used to solve this problem, provided that $M \geq K+1$. For an in-depth discussion of subspace cisoid estimation algorithms we refer the reader to~\cite{Stoica2005}. In the presence of noise, the success of subspace methods is highly dependent on the condition number of the Vandermonde matrix with poles $\{z_k\}_{k = 1}^K$ \cite{Potts2013}. In \cite{Bazan2000}, it was shown that these Vandermonde matrices are relatively well conditioned, provided that the nodes lie inside the unit disk, remain close to the unit circle, but are not extremely close to each other.

\section*{Acknowledgment}
The authors would like to thank David Stotz and Reinhard Heckel for inspiring discussions, as well as one of the reviewers for helpful comments.



%
\bibliographystyle{IEEEtran} 
\bibliography{ref}

\end{document}